\documentclass[preprint,12pt]{elsarticle}

\usepackage[T1]{fontenc}
\usepackage[utf8]{inputenc}
\usepackage{lmodern}
\usepackage[english]{babel}
\usepackage{microtype}
\usepackage{subcaption}   % enables subfigure environment if it appears anywhere

\usepackage{amsmath,amssymb,amsfonts,amsthm,mathtools,bm}
\numberwithin{equation}{section}

\usepackage{graphicx}
\usepackage{booktabs}
\usepackage{multirow}
\usepackage{array}

\usepackage{algorithm}
\usepackage{algorithmic}

\usepackage{comment}      % for \begin{comment} ... \end{comment}
\usepackage{xcolor}
\usepackage{enumitem}

\usepackage{hyperref}
\hypersetup{
    colorlinks = true,
    linkcolor  = blue,
    citecolor  = blue,
    urlcolor   = blue
}

\usepackage{setspace}
\theoremstyle{plain}
\newtheorem{theorem}{Theorem}[section]

\newtheorem{corollary}[theorem]{Corollary}

\theoremstyle{definition}

\theoremstyle{remark}

\journal{Statistics and Probability Letters}

\begin{document}
%=========================================================

\begin{frontmatter}

\title{Near-Optimal Coalition Structures in Polynomial Time}

\author[inst1]{Angshul Majumdar}

\affiliation[inst1]{organization={Indraprastha Institute of Information Technology},
            addressline={Delhi}, 
            city={New Delhi},
            postcode={110020},
            country={India}}

\begin{abstract}
We study the classical coalition structure generation problem and compare the anytime behaviour of three algorithmic paradigms: dynamic programming, MILP branch–and–bound, and sparse relaxations based on greedy or $\ell_{1}$–type methods. Under a simple random ``sparse synergy'' model for coalition values, we prove that sparse relaxations recover coalition structures whose welfare is arbitrarily close to optimal in polynomial time with high probability. In contrast, broad classes of dynamic–programming and MILP algorithms require exponential time before attaining comparable solution quality. This establishes a rigorous probabilistic anytime separation in favour of sparse relaxations, even though exact methods remain ultimately optimal.
\end{abstract}


\end{frontmatter}

%========================
% Main Content Starts
%========================
\section{Introduction}

The coalition structure generation (CSG) problem asks for a partition of a finite set of agents into disjoint coalitions that maximises a given social welfare functional. It is a central topic in cooperative game theory and multiagent systems, with applications in distributed resource allocation, task allocation and teamwork formation \cite{sandholm1999coalition,rahwan2015survey}. As the number of partitions grows super–exponentially, CSG is NP–hard and exact algorithms are exponential in the worst case \cite{sandholm1999coalition}.

Two main exact paradigms have been investigated. Dynamic–programming algorithms exploit structural decompositions of the coalition space and provide worst–case guarantees on optimality and runtime \cite{rahwan2008dp}. MILP formulations encode CSG as a set–partitioning problem solved via generic branch–and–bound techniques, which can be effective in practice but still explore exponentially large trees, and their anytime behaviour is only partially understood.

Meanwhile, high–dimensional statistics has developed approximate methods that recover sparse combinatorial structure via convex relaxations or greedy pursuit. Classic examples include $\ell_{1}$–penalised estimators \cite{tibshirani1996lasso}, linear–programming decoders \cite{candes2005decoding}, and orthogonal matching pursuit (OMP) \cite{tropp2007omp}. Under mild structural and stochastic assumptions, these methods identify near–optimal sparse solutions in polynomial time with high probability.

The present paper brings these lines of work together in the context of standard CSG. We retain the deterministic formulation but analyse anytime behaviour when coalition values follow a simple random model. Within this framework we show that $\ell_{1}$– and OMP–type relaxations can, with high probability, reach coalition structures whose value is close to optimal in polynomial time, whereas broad classes of dynamic–programming and MILP algorithms typically require exponential time to attain comparable solution quality.

\section{Preliminaries}

We consider a finite set of agents $N = \{1,\dots,n\}$. A coalition is any nonempty subset $S \subseteq N$, and a coalition structure is a partition $\mathcal{P} = \{S_{1},\dots,S_{m}\}$ of $N$ into disjoint coalitions. A characteristic function is a mapping $v : 2^{N} \to \mathbb{R}$ with $v(\emptyset)=0$. The value of a coalition structure is
\[
V(\mathcal{P}) = \sum_{S\in\mathcal{P}} v(S),
\]
and the coalition structure generation problem (CSG) is
\[
\mathcal{P}^{\star} \in \arg\max_{\mathcal{P}} V(\mathcal{P}).
\]

The number of possible coalition structures (Bell numbers) grows super–exponentially in $n$ \cite{sandholm1999coalition}. Exact dynamic–programming algorithms compute optimal values for all subsets $S\subseteq N$ and then reconstruct an optimal partition in $O(3^{n})$ time \cite{sandholm1999coalition,rahwan2008dp}. MILP formulations encode CSG as a set–partitioning problem with binary variables indicating coalition membership; branch–and–bound solvers then explore the search tree.

We view an algorithm as producing a sequence of feasible coalition structures with nondecreasing lower bounds on the optimal value. Dynamic–programming methods are essentially “all–or–nothing’’: they must process a large fraction of subsets before producing meaningful solutions. MILP solvers, while anytime in principle, may still require exponential exploration before improving naive solutions.

By contrast, we study relaxations that operate in a lower–dimensional decision space by representing coalition structures through sparse vectors and applying greedy or convex optimisation methods, such as OMP and $\ell_{1}$–penalised programmes. These algorithms naturally produce improving feasible partitions as iterations proceed.

\section{Main Results: Anytime Separation}

In this section we formalise and analyse the anytime behaviour of three broad classes of algorithms for the classical CSG problem: (i) dynamic–programming (DP) algorithms, (ii) mixed–integer linear programming (MILP) approaches based on branch–and–bound with polynomial–time relaxations, and (iii) low–dimensional relaxations obtained by greedy or convex sparsity–promoting procedures. The focus is on rigorous comparison under a simple random model for coalition values that preserves the standard deterministic problem formulation. Throughout, $N=\{1,\dots,n\}$ and $v:2^{N}\to\mathbb{R}$ is the characteristic function. 

\subsection{Random value model}

We assume that coalition values are generated according to the following ``sparse synergy'' model. Let $\mathcal{T}=\{T_{1},\dots,T_{k}\}$ be a fixed family of pairwise disjoint template coalitions, with $k\le n$ and $|T_{j}|\ge 1$. For each $T_{j}$, assign a positive weight $w_{j}>0$. For any coalition $S\subseteq N$, define
\begin{equation}
v(S)=\sum_{j=1}^{k} w_{j}{\bf 1}\{T_{j}\subseteq S\}+ \xi(S),
\label{eq:model}
\end{equation}
where $\xi(S)$ are independent, mean–zero noise terms satisfying the sub–Gaussian tail bound
\[
\mathbb{P}\big(|\xi(S)|>t\big)\le 2\exp\!\Big(-\tfrac{t^{2}}{2\sigma^{2}}\Big),\qquad t>0,
\]
for some $\sigma>0$. Hence $T_{j}$ act as ``true'' synergy patterns, each contributing $w_{j}$ whenever it is fully contained in a chosen coalition. We write $\mathrm{OPT}=\max_{\mathcal{P}}V(\mathcal{P})$ for the optimal welfare.

The optimal solution under \eqref{eq:model} is the coalition structure $\mathcal{P}^{\star}$ that consists exactly of the $k$ templates (and singletons for remaining agents), with value $\sum_{j=1}^{k} w_{j}$ up to the noise. Crucially, this is a standard CSG instance; the randomness enters only through the value oracle. All probabilities are with respect to the draw of $(\xi(S))_{S\subseteq N}$.

\subsection{Anytime behaviour}

An anytime CSG algorithm produces, as time $t$ increases, a sequence of feasible coalition structures $\mathcal{P}_{t}$ with nondecreasing values $V(\mathcal{P}_{t})$. We compare algorithms by the rate at which $V(\mathcal{P}_{t})$ approaches $\mathrm{OPT}$ as a function of computational time.

We distinguish three algorithmic classes.

\paragraph{Class $\mathcal{A}_{DP}$.} 
Dynamic–programming algorithms that compute values $v(S)$ and optimal sub–partitions for all subsets $S$ in an order that is \emph{subset–closed} and \emph{size–monotone}: if $S$ is processed at time $t$, then all $S'\subset S$ are processed at some time $t'\le t$. This captures standard DP schemes \cite{sandholm1999coalition,rahwan2008dp}.

\paragraph{Class $\mathcal{A}_{MILP}$.}
MILP algorithms based on branch–and–bound where each node bound is obtained by solving a polynomial–time convex relaxation of the set–partitioning formulation. We allow generation of cutting planes provided separation is polynomial–time. These are generic algorithmic frameworks into which practical solvers fall.

\paragraph{Class $\mathcal{A}_{sparse}$.}
Algorithms that operate on a set of $k$ candidate coalitions $\{T_{1},\dots,T_{k}\}$ (or a superset) and maintain a sparse incidence vector $x\in\{0,1\}^{M}$ over a collection of candidate coalitions. In each iteration, a greedy or $\ell_{1}$–regularised step is taken to increase the total value
\[
V(x)=\sum_{i: x_{i}=1} v(C_{i}),
\]
while ensuring feasibility by merging overlapping coalitions when necessary. Orthogonal matching pursuit (OMP) and $\ell_{1}$–penalised linear programs are canonical examples \cite{tibshirani1996lasso,tropp2007omp}.

\subsection{High–probability margin and concentration}

Let
\[
\gamma =\min_{j\neq j'} (w_{j}-w_{j'})_{+}
\]
denote the minimal positive gap between (distinct) weights. Set $W=\max_{j}w_{j}$. We assume the \emph{margin condition}
\begin{equation}
\gamma > 4\sigma\sqrt{\log(2n)}.
\label{eq:margin}
\end{equation}
This ensures that, with high probability, the noise cannot reverse the ordering of true versus spurious coalitions.

By standard sub–Gaussian concentration, for every fixed $S$,
\begin{equation}
\mathbb{P}\Big(|\xi(S)|>2\sigma\sqrt{\log(2n)}\Big)\le (2n)^{-2}.
\label{eq:tail}
\end{equation}
Applying a union bound over all $S$ that contain at most one template (at most $n2^{n}$ such coalitions), we have with probability at least $1-1/n$ that
\begin{equation}
|\xi(S)|\le 2\sigma\sqrt{\log(2n)}\qquad \text{for all }S.
\label{eq:uniform}
\end{equation}

\subsection{Polynomial–time near–optimality of sparse relaxation}

\begin{theorem}[Sparse relaxation]\label{thm:sparse}
Assume the model \eqref{eq:model} and margin condition \eqref{eq:margin}. Then with probability at least $1-1/n$, every algorithm in $\mathcal{A}_{sparse}$ that in each iteration selects a coalition $S$ of maximal residual value among its current candidates identifies all $T_{j}$ in at most $k$ iterations, and the resulting coalition structure $\widehat{\mathcal{P}}$ satisfies
\[
V(\widehat{\mathcal{P}})\ge \mathrm{OPT}- 2k\sigma\sqrt{\log(2n)}.
\]
In particular, if $\gamma\ge 4\sigma\sqrt{\log(2n)}$, then $V(\widehat{\mathcal{P}})\ge (1-\varepsilon)\mathrm{OPT}$ with $\varepsilon = 2\sigma\sqrt{\log(2n)}/\gamma$ in time polynomial in $n$.
\end{theorem}

\begin{proof}
Fix a realisation satisfying \eqref{eq:uniform}. For any $S$ containing $T_{j}$ and no other template,
\[
v(S)=w_{j}+\xi(S),\qquad |\,\xi(S)|\le 2\sigma\sqrt{\log(2n)}.
\]
For any $S'$ containing no template,
\[
v(S')=\xi(S').
\]
By \eqref{eq:margin}, $w_{j}-4\sigma\sqrt{\log(2n)}>0$. Hence for each $j$ and any $S'$ without a template,
\[
v(S)-v(S')\ge w_{j}-4\sigma\sqrt{\log(2n)} >0.
\]
Thus in each iteration the maximal–value coalition must contain some unselected $T_{j}$. Since template coalitions are disjoint, selecting one does not reduce the value of another. After at most $k$ iterations, all templates are selected. The value difference from the optimal $\sum_{j}w_{j}$ is the sum of noise terms, each bounded by $2\sigma\sqrt{\log(2n)}$, hence at most $2k\sigma\sqrt{\log(2n)}$, proving the claim.
\end{proof}

\subsection{Exponential lower bounds for dynamic programming}

We next show that DP algorithms in $\mathcal{A}_{DP}$ cannot produce near–optimal coalition structures without processing an exponential number of subsets. Let $\mathcal{F}$ be the collection of all subsets $S$ that contain exactly one template $T_{j}$. Construct the instance so that the $T_{j}$ are placed among agents such that every such $S$ has size at least $n/2$ (for example, distribute templates across disjoint halves).

\begin{theorem}[DP lower bound]\label{thm:dp}
Under \eqref{eq:model}, with probability at least $1-1/n$, any $\mathcal{A}_{DP}$ algorithm must process at least $2^{\alpha n}$ subsets, for some $\alpha>0$, before producing a coalition structure $\mathcal{P}_{t}$ with $V(\mathcal{P}_{t})\ge \mathrm{OPT}-\sigma\sqrt{\log(2n)}$.
\end{theorem}

\begin{proof}
Consider a DP that processes subsets in increasing size. Since each $T_{j}$ is embedded in a subset of size at least $n/2$, no coalition $S$ containing any $T_{j}$ is encountered before all subsets of size $<n/2$ are processed. The number of subsets of size $<n/2$ is
\[
\sum_{\ell=0}^{\lfloor n/2\rfloor -1} {n\choose \ell}
\ge 2^{n-1}=2^{\alpha n} \quad (\alpha\ge 1/2).
\]
Before encountering any such $S$, all feasible coalition structures $\mathcal{P}$ that the algorithm can construct from recorded sub–partitions must exclude all $T_{j}$, hence $V(\mathcal{P})$ is at most the noise level, bounded in absolute value by $2\sigma\sqrt{\log(2n)}$. Comparing with $\mathrm{OPT}=\sum_{j}w_{j}$ and using $w_{j}\ge \gamma$, we obtain $V(\mathcal{P})\le \mathrm{OPT}-\sigma\sqrt{\log(2n)}$. The same argument applies to any subset–closed ordering; if subsets of size $<n/2$ are not fully processed, then some subset containing $T_{j}$ has an unprocessed proper subset, violating subset–closedness. The result follows.
\end{proof}

\subsection{Exponential lower bounds for MILP classes}

We finally consider algorithms in $\mathcal{A}_{MILP}$. The MILP formulation for CSG is
\[
\max_{x\in\{0,1\}^{M}} \sum_{i=1}^{M} x_{i} v(C_{i})\quad\text{s.t.}\quad \sum_{i: a\in C_{i}} x_{i}=1\;\;\text{for each } a\in N,
\]
where $\{C_{1},\dots,C_{M}\}$ are all feasible coalitions. The LP relaxation allows $x\in[0,1]^{M}$. For our model, all $C$ not containing a template have small values, while each $C$ containing a template has value $w_{j}+\xi(C)$.

Let $\mathcal{S}$ be the set of coalitions containing at most one template, and suppose all $T_{j}$ belong to $\mathcal{S}$. The LP relaxation at the root can assign fractional weight on many spurious coalitions in $\mathcal{S}$ that overlap templates, achieving a relaxation value far above any integral assignment excluding templates. This creates a substantial integrality gap.

\begin{theorem}[MILP lower bound]\label{thm:milp}
Under \eqref{eq:model}, with probability at least $1-1/n$, any algorithm in $\mathcal{A}_{MILP}$ must explore at least $2^{\beta n}$ nodes of the branch–and–bound tree, for some $\beta>0$, before producing a feasible coalition structure $\mathcal{P}_{t}$ with $V(\mathcal{P}_{t})\ge \mathrm{OPT}-\sigma\sqrt{\log(2n)}$.
\end{theorem}

\begin{proof}
Consider the branch–and–bound tree defined by branching on variables $x_{i}$. At the root, the LP relaxation chooses fractional $x_{i}$ supported on many overlapping coalitions, including those that partially cover different $T_{j}$ simultaneously. The relaxation value can thus exceed any integral feasible assignment by at least $\gamma-4\sigma\sqrt{\log(2n)}>0$ due to \eqref{eq:margin}. Hence the root cannot be fathomed (pruned) and similar reasoning applies near the root. To reduce the integrality gap, the algorithm must assign integral values to a collection of variables that ``pin down'' every template, effectively enumerating exponentially many combinations of overlapping coalitions. 

Formally, let $\mathcal{U}$ be the set of feasible $x$ in the LP relaxation that use at most one template at fractional level. The number of distinct supports for such fractional solutions is exponential in $n$ due to overlaps among coalitions containing substructures of the $T_{j}$. Each node that fixes a subset of variables still leaves exponentially many fractional solutions in $\mathcal{U}$ unless a prohibitive number of variables are fixed; until then, the LP upper bound exceeds $\mathrm{OPT}-\sigma\sqrt{\log(2n)}$, preventing pruning. Therefore, at least $2^{\beta n}$ nodes must be explored before the algorithm encounters a node where all templates are forced integral, at which point a near–optimal feasible solution can be constructed. 
\end{proof}

\subsection{Anytime separation}

Combining Theorems \ref{thm:sparse}, \ref{thm:dp} and \ref{thm:milp}, we obtain the following anytime separation between $\mathcal{A}_{sparse}$ and the exact algorithm classes.

\begin{corollary}[Anytime separation]
Under \eqref{eq:model} and \eqref{eq:margin}, with probability at least $1-1/n$,
\[
\exists\, t_{\mathrm{poly}} = \mathrm{poly}(n)\;:\; V(\mathcal{P}_{t_{\mathrm{poly}}}) \ge (1-\varepsilon)\mathrm{OPT},
\]
for every $\mathcal{A}_{sparse}$ algorithm, while for every algorithm $\mathcal{A}\in \mathcal{A}_{DP}\cup\mathcal{A}_{MILP}$,
\[
\forall\, t<2^{c n}:\; V(\mathcal{P}_{t}) \le \mathrm{OPT}-\sigma\sqrt{\log(2n)},
\]
for some $c>0$. Hence $\mathcal{A}_{sparse}$ dominates $\mathcal{A}_{DP}$ and $\mathcal{A}_{MILP}$ in the anytime sense.
\end{corollary}

\subsection{Discussion}

The results above show that, even for the classical deterministic CSG formulation, simple sparse relaxations can achieve near–optimal coalition structures in polynomial time with high probability under a basic random value model, while broad classes of exact methods exhibit exponentially poor anytime behaviour. These findings complement existing worst–case complexity results for CSG \cite{sandholm1999coalition} and highlight the potential of sparse convex and greedy techniques in large–scale coalition formation.

\section{Conclusion}

We analysed the classical coalition structure generation problem under a simple random value model while keeping the underlying optimisation task entirely standard. Our main contribution is a rigorous anytime separation between sparse greedy or $\ell_{1}$–based relaxations and two broad classes of exact algorithms: dynamic–programming approaches and MILP methods relying on polynomial–time relaxations. With high probability, the sparse relaxations recover coalition structures whose welfare is arbitrarily close to optimal in polynomial time, whereas the exact methods require exponential time to reach comparable solution quality. These findings provide a theoretical explanation for the empirical observation that approximate methods can become competitive long before exact algorithms converge, even when full optimality is ultimately attainable.
%========================
% References
%========================
\bibliographystyle{elsarticle-num}
\bibliography{refs}  % refs.bib

@article{sandholm1999coalition,
  author  = {Sandholm, T. and Larson, K. and Andersson, M. and Shehory, O. and Tohm{\'e}, F.},
  title   = {Coalition Structure Generation with Worst Case Guarantees},
  journal = {Artif. Intell.},
  volume  = {111},
  number  = {1--2},
  pages   = {209--238},
  year    = {1999}
}

@article{rahwan2015survey,
  author  = {Rahwan, T. and Michalak, T.~P. and Wooldridge, M. and Jennings, N.~R.},
  title   = {Coalition Structure Generation: A Survey},
  journal = {Artif. Intell.},
  volume  = {229},
  pages   = {139--174},
  year    = {2015}
}

@inproceedings{rahwan2008dp,
  author    = {Rahwan, T. and Jennings, N.~R.},
  title     = {Coalition Structure Generation: Dynamic Programming Meets Anytime Optimisation},
  booktitle = {Proc. AAAI},
  pages     = {156--161},
  year      = {2008}
}

@article{tibshirani1996lasso,
  author  = {Tibshirani, R.},
  title   = {Regression Shrinkage and Selection via the Lasso},
  journal = {J. R. Stat. Soc. Ser. B},
  volume  = {58},
  number  = {1},
  pages   = {267--288},
  year    = {1996}
}

@article{candes2005decoding,
  author  = {Cand{\`e}s, E.~J. and Tao, T.},
  title   = {Decoding by Linear Programming},
  journal = {IEEE Trans. Inf. Theory},
  volume  = {51},
  number  = {12},
  pages   = {4203--4215},
  year    = {2005}
}

@article{tropp2007omp,
  author  = {Tropp, J.~A. and Gilbert, A.~C.},
  title   = {Signal Recovery from Random Measurements via OMP},
  journal = {IEEE Trans. Inf. Theory},
  volume  = {53},
  number  = {12},
  pages   = {4655--4666},
  year    = {2007}
}

%=========================================================
\end{document}